\documentclass[letterpaper, 10 pt, conference]{ieeeconf}  
\IEEEoverridecommandlockouts                              
\usepackage{graphics} 
\usepackage{epsfig} 
\usepackage[noadjust]{cite}

\usepackage{amsmath,amssymb,amsfonts,bm}

\usepackage{amsthm}
\usepackage[ruled]{algorithm2e}
\usepackage{xcolor}
\usepackage{tikz}
\usepackage[utf8]{inputenc}
\usepackage{pgfplots}
\usetikzlibrary{shapes,arrows,fit,positioning,calc}
\usetikzlibrary{decorations.pathreplacing}
\usetikzlibrary{arrows.meta, decorations.pathmorphing}
\usepgfplotslibrary{groupplots}
\usepackage{booktabs}
\usepackage{diagbox}
\usepackage{hyperref}
\usepackage{threeparttable}

\newcommand{\score}{scoring function}

\newcommand{\Ib}{{\mathbb{I}}}

\newcommand{\Yc}{{\mathcal{Y}}}

\newcommand{\Uc}{{\mathcal{U}}}

\newcommand{\Hc}{{\mathcal{H}}}

\newcommand{\revise}[1]{\textcolor{black}{#1}}

\newtheoremstyle{mystyle}
  {}
  {}
  {\upshape}
  {}
  {\bfseries}
  {.}
  { }
  {}

\newtheorem{proposition}{Proposition}
\theoremstyle{mystyle}
\newtheorem{remark}{Remark}
\newtheorem{problem}{Problem}

\title{Learning-Based Efficient Approximation of Data-Enabled Predictive Control
}

\author{Yihan Zhou$^{1}$, Yiwen Lu$^{1}$, Zishuo Li$^{1}$, Jiaqi Yan$^{2}$, Yilin Mo$^{1}$
\thanks{This work was supported by the National Natural Science Foundation of China (NSFC) under grant number 62192752 and number 62273196.}
\thanks{$^{1}$Yihan Zhou, Yiwen Lu, Zishuo Li, and Yilin Mo are with the Department of Automation and BNRist, Tsinghua University, Beijing, 100084, China.
({\small \{zhouyh23, luyw20, lizs19\}@mails.tsinghua.edu.cn, ylmo@tsinghua.edu.cn})}
\thanks{$^{2}$Jiaqi Yan is with the Automatic Control Laboratory, ETH Zurich, Switzerland. ({\small jiayan@ethz.ch})}
}

\begin{document}

\maketitle

\begin{abstract}

Data-Enabled Predictive Control (DeePC) bypasses the need for system identification by directly leveraging raw data to formulate optimal control policies. However, the size of the optimization problem in DeePC grows linearly with respect to the data size, which prohibits its application \revise{to resource-constrained systems} due to high computational costs.
In this paper, we propose an efficient approximation of DeePC, whose size is invariant with respect to the amount of data collected, via differentiable convex programming. Specifically, the optimization problem in DeePC is decomposed into two parts: a control objective and a scoring function that evaluates the likelihood of a guessed I/O sequence, the latter of which is approximated with a size-invariant learned optimization problem. The proposed method is validated through numerical simulations on a quadruple tank system, illustrating that the learned controller can reduce the computational time of DeePC by a factor of 5 while maintaining its control performance.

\end{abstract}

\section{INTRODUCTION}

Optimal trajectory tracking has long been a fundamental problem in control systems. Classical model-based control approaches, such as Model Predictive Control (MPC), typically rely on a \emph{state-space model} of the system to predict future trajectories and determine the optimal control decision. However, defining the state space model can be a laborious part of the control design~\cite{hjalmarsson2005experiment,ogunnaike1996contemporary}, and even after it is defined, identifying its parameters from noisy observations is known to be difficult~\cite{tsiamis2021linear,li2022fundamental}.

In recent years, direct data-driven control~\cite{de2019formulas,dorfler2022bridging,zhao2024data,WANG2025111897} has emerged as a promising solution to the aforementioned challenges, with Data-Enabled Predictive Control (DeePC)~\cite{coulson2019data} being a representative algorithm encompassing this concept. A key idea of direct data-driven control is to construct a \emph{non-parametric model} of the system using a data matrix formed by Input/Output (I/O) trajectories of the system collected \textit{in advance}. According to \revise{Willems'} fundamental lemma~\cite{willems1997introduction}, any feasible I/O trajectory of a noise-free linear system must be spanned by these offline trajectories, as long as the persistency of excitation condition~\cite{willems_note_2005} is satisfied, which justifies a non-parametric representation. Based on this non-parametric model, DeePC tackles the tracking problem by solving an optimization problem in a receding-horizon manner similarly to MPC, but with the non-parametric model replacing the state-space model for prediction.

DeePC has received significant attention from both theoretical~\cite{berberich2020data,baros2022online} and empirical~\cite{carlet_data-driven_2022, elokda2021data} perspectives. However, a notable challenge associated with the practical deployment of DeePC lies in its computation demand~\cite{dorfler2022bridging, fiedler_relationship_2021,huang2019data}. Since the I/O data collected from real-world systems are typically contaminated with noise, a large amount of data is required to generate an accurate non-parametric model of the system. However, as the amount of data used for building the nonparametric model increases, width of the data matrix grows accordingly, resulting in an increase in the dimension of a decision variable in DeePC.
Therefore, the computational burden of DeePC scales with the amount of data collected, which can be prohibitive for real-time applications such as robotics and power electronics~\cite{wegner_data-enabled_2021}.

Previous research efforts have been devoted to mitigating the computational challenge of DeePC by compressing the data matrix. On the one hand, for lossless compression, a representation of predicted trajectories based on the null-space of the data matrix is proposed~\cite{carlet_data-driven_2022}. On the other hand, approximation methods have been developed to perform lossy compression of the data matrix, including those based on Singular Value Decomposition (SVD)~\cite{zhang_dimension_2023}, Proper Orthogonal Decomposition (POD)~\cite{carlet_real-time_2023}, and LQ decomposition~\cite{breschi2023data, sader_causality-informed_2023}. Overall, these methods amount to diminishing the number of effective trajectory segments in the nonparametric model, i.e., reducing the width of the data matrix.

In this paper, we propose an alternative view of DeePC: we show that the DeePC problem is equivalent to minimizing the sum of a control objective $\ell(\tau)$ and an anomaly metric $S(\tau)$, both as functions of the predicted I/O sequence $\tau$. The former encompasses the tracking error, the penalty on control effort, and the constraints on I/O signals, which is independent from the data collected. The latter maps an I/O sequence to a scalar representing how unlikely $\tau$ is a trajectory generated by the system, whose evaluation is based on the data. We dub $S(\tau)$ as a \emph{scoring function} or \emph{scoring model}, since it can serve as an oracle whose queries return a score indicating the fitness of a guessed I/O sequence to the system dynamics. Although this reformulation is equivalent to the original DeePC, it motivates us to accelerate computation of the controller via a learning-based approximation $\hat{S}(\tau)$ of the scoring function whose computational complexity is fixed even with more data collected. Specifically, we propose to parameterize the approximate scoring function as a differentiable convex program, and train its parameters via supervised learning against the true scoring function $S$.
With our proposed learning objective based on proximal operators, the minimizer of $\ell + \hat{S}$ is close to the minimizer of $\ell + S$, indicating that the approximate scoring function can be used to solve the overall control problem effectively.

An advantage of our approach is its computational merit in decoupling the scale of the control problem from the amount of data used. That is, once the form of the approximate scoring function is determined, one can train its parameters \emph{offline} with a large dataset, without affecting the number of variables or constraints in the control problem to be solved \emph{online}. This stands in contrast with the original DeePC, which, as mentioned, has a set of variables whose dimension scales with the amount of data. Hence, the learned approximate scoring function $\hat{S}$ can be viewed as a condensed representation of the system dynamics, which, combined with the control objective, enables computationally efficient control. The above intuitions are supported by numerical simulations, which show that our proposed controller can achieve a similar control performance to DeePC in less computational time.

Our contributions can be summarized as follows: i) We present a reformulation of DeePC based on the notion of scoring model, contributing to a new perspective on the data-enabled control problem. ii) We propose an approximate DeePC control law derived from a learning-based approximation of the scoring function. iii) We demonstrate through numerical simulations that the proposed method can significantly reduce the computational time required for DeePC, with only a minor compromise in control performance.

The rest of the paper is organized as follows. Section II provides a brief overview of the DeePC algorithm and its computational restriction. Section III presents our learning-based approximation of DeePC. Section IV presents the simulation results. Section V concludes the paper.

\emph{Notations:}
The set of all integers is denoted by $\mathbb{Z}$, with subsets of positive integers and non-negative integers represented as $\mathbb{Z}_{>0}$ and $\mathbb{Z}_{\geq0}$, respectively. In the context of matrices and vectors, $X^{\top}$ denotes the transpose of matrix $X$, and $X^{\dagger}$ represents the Moore-Penrose pseudo-inverse of matrix $X$. The operator $\operatorname{diag}(x_1, \ldots, x_n)$ creates a diagonal matrix with elements $x_1, \ldots, x_n$ on its diagonal. The indicator function $\mathbb{I}_{\mathcal{X}}(x)$ equals 0 if $x \in \mathcal{X}$ and $+\infty$ otherwise. The subdifferential and proximal operator of a function $f$ are denoted by $\partial f$ and $\operatorname{Prox}_f$, respectively. The operator $\operatorname{col}$ is used to denote the column stack of vectors. $I$ denotes the identity matrix of appropriate dimension or the identity operator according to the context. The $p$-norm of a vector or matrix or operator is denoted by $\|\cdot\|_p$.
The quadratic form $\|x\|_Q^2 = x^{\top}Qx$ signifies the weighted norm squared of vector $x$.
The symbol $\otimes$ represents the Kronecker product. The symbol $\mathcal{N}(\mu, \Sigma)$ denotes a Gaussian distribution with mean $\mu$ and covariance $\Sigma$.
\section{PRELIMINARIES}

\subsection{Problem Formulation}

In this paper, we consider a discrete-time Linear Time-Invariant (LTI) system with state-space representation
\begin{subequations}\label{eq:state_space}
\begin{align}
   &x(t+1)=Ax(t)+Bu(t)+w(t), \\
   &y(t)=Cx(t)+v(t),
\end{align}
\end{subequations}
where $A \in \mathbb{R}^{n \times n}, B \in \mathbb{R}^{n \times m}, C \in \mathbb{R}^{p \times n}$ are state-space matrices, $x(t)\in \mathbb{R}^n, u(t)\in \mathbb{R}^m, y(t)\in \mathbb{R}^p$ are state, input, and output vectors at time $t \in \mathbb{Z}_{\geq 0}$, respectively, and $w(t) \stackrel{\text{i.i.d.}}{\sim} \mathcal{N}(0, \Sigma_w)$ and $v(t) \stackrel{\text{i.i.d.}}{\sim} \mathcal{N}(0, \Sigma_v)$ are process and measurement noises. We assume that $(A, B)$ is controllable and $(A, C)$ is observable.

We consider an optimal tracking problem of the system~\eqref{eq:state_space}, whose objective is to minimize a quadratic function
\begin{equation}
   \limsup_{T\to\infty} \frac{1}{T} \sum_{t = 0}^{T-1} \left( \| y(t) - r \|_Q^2 + \| u(t) \|_R^2 \right),
   \label{eq:objective}
\end{equation}
subject to the constraints
\begin{equation}
   y(t) \in \mathcal{Y}, u(t) \in \mathcal{U}, \forall t \in \mathbb{Z}_{\geq 0},
   \label{eq:constraints}
\end{equation}
where $r$ is a reference signal to be tracked, $\mathcal{Y} \subseteq \mathbb{R}^p$ and $\mathcal{U} \subseteq \mathbb{R}^m$ are convex sets representing the output and input constraints respectively, and \revise{$Q \in \mathbb{R}^{p \times p} ,R \in \mathbb{R}^{m \times m}$} are symmetric positive definite weighting matrices.

\begin{remark}
   The reference signal can be a time-varying signal $r(t)$ for all of the methods in this paper, but we assume that $r$ is a constant for notational simplicity.
\end{remark}

We consider a setting where the state-space representation~\eqref{eq:state_space} is unavailable to the system operator, i.e., none of the state-space matrices $A, B, C$, the state vectors $x(t)$, or the state dimension $n$ is known. Instead, only the input and output vectors $u(t)$ and $y(t)$ are observed.

\subsection{Data-Enabled Predictive Control (DeePC)}

DeePC~\cite{coulson2021distributionally} tackles the aforementioned optimal tracking problem in a receding horizon manner, where the optimization problem described below is solved at every time step $t$. We use the convention that the index in the parentheses, e.g., $y(t)$, denotes the actual value of the variable at time $t$., while the index in subscript, e.g., $y_k$, denotes the predicted value of the variable at the $k$-th step in the \revise{prediction} horizon.

\begin{problem}[DeePC Problem at Time Step $t$]  \label{eq:DeePC}
   \begin{align}
         & \underset{u_{0:\!N\!-\!1}, y_{0:\!N\!-\!1}, g, \sigma_y}\min  \sum_{k=0}^{N-1}\left(\left\|y_k-r\right\|_Q^2+\left\|u_k\right\|_R^2\right) +\notag \\
         & \qquad  \lambda_{g1}\|g\|_1 + \lambda_{g2}\|g\|_2^2  + \lambda_{y1}\left\|\sigma_y\right\|_1 + \lambda_{y2}\left\|\sigma_y\right\|_2^2 \notag \\[2ex]
         & \mathrm { s.t. }
      \begin{array}{r}
      \textcolor{black}{\scriptstyle mT_{\mathrm{ini}}\{} \\
      \textcolor{black}{\scriptstyle mN\{} \\
      \textcolor{black}{\scriptstyle pT_{\mathrm{ini}}\{} \\
      \textcolor{black}{\scriptstyle pN\{}
      \end{array}
         \underbrace{
         \left[\begin{array}{c}
         U_{\mathrm{p}} \\
         U_{\mathrm{f}} \\
         Y_{\mathrm{p}} \\
         Y_{\mathrm{f}}
         \end{array}\right]
         }_{\scriptstyle M}
         g=\left[\begin{array}{c}
         u_{\mathrm{ini}} \\
         u_{0:N-1} \\
         y_{\mathrm{ini}} \\
         y_{0:N-1}
         \end{array}\right]+\left[\begin{array}{c}
         0 \\
         0 \\
         \sigma_y \\
         0
         \end{array}\right], \notag \\
         & \mkern90mu u_{k} \in \mathcal{U}, \forall k \in\{0, \ldots, N-1\}, \label{eq:def_Uc}\\
         & \mkern90mu y_{k} \in \mathcal{Y}, \forall k \in\{0, \ldots, N-1\},\label{eq:def_Yc}
   \end{align}
   where the new notations are defined and explained as follows: \begin{itemize}
      \item $N \in \mathbb{Z}_{>0}$ is the prediction horizon, $M \in \mathbb{Z}_{>0}$ is the number of trajectory segments collected offline, and $T_{\mathrm{ini}} \geq l$ is the length of the initialization sequence, where $l$ is the lag~\cite[Section 7.2]{markovsky2006exact} of the system.
      \item The matrix $\mathcal{H}\triangleq \begin{bmatrix}
         U_p^\top & U_f^\top & Y_p^\top & Y_f^\top
      \end{bmatrix}^\top$ is the data matrix, each column of which is formed by concatenating the Input/Output (I/O) signals collected offline at $T_{\mathrm{ini}} + N$ consecutive time steps. We call such a column a \emph{trajectory segment}. The data matrix has $M$ columns, and is partitioned into row block matrices $U_{\mathrm{p}}, U_{\mathrm{f}}, Y_{\mathrm{p}}, Y_{\mathrm{f}}$, with number of rows $mT_{\mathrm{ini}}, mN, pT_{\mathrm{ini}}, pN$ respectively, where the subscripts $\mathrm{p}$ and $\mathrm{f}$ stand for past and future, respectively. It is assumed that the matrix $\begin{bmatrix}
         U_{\mathrm{p}}^\top & U_{\mathrm{f}}^\top
      \end{bmatrix}^\top$ is of full row rank.
      \item $u_{\mathrm{ini}} \in \mathbb{R}^{mT_{\mathrm{ini}}}$ and $y_{\mathrm{ini}} \in \mathbb{R}^{pT_{\mathrm{ini}}}$ are column vectors formed by concatenating the I/O signals of the past $T_{\mathrm{ini}}$ time steps, which are inputs to the DeePC problem.
      \item $u_{0:\!N\!-\!1} \in \mathbb{R}^{mN}, y_{0:\!N\!-\!1} \in \mathbb{R}^{pN}$ stand for the predicted I/O sequence over the horizon $N$, which are optimization variables of the DeePC problem. After solving the DeePC problem, $u_0^\star$ from the optimal solution is applied to the system, i.e., $u(t) = u_0^\star$.
      \item $g \in \mathbb{R}^{M}$ and $\sigma_y \in \mathbb{R}^{pN}$ are auxiliary optimization variables of the DeePC problem, which represent the coefficients for spanning the predicted I/O trajectory sequence from the data matrix and the slack variable for the output constraints, respectively. The scalar parameters $\lambda_{g1}, \lambda_{g2}, \lambda_{y1}, \lambda_{y2} \in \mathbb{R}_{\geq 0}$ are predefined coefficients. When brought together, $\lambda_{g1}\|g\|_1 + \lambda_{g2}\|g\|_2^2  + \lambda_{y1}\left\|\sigma_y\right\|_1 + \lambda_{y2}\left\|\sigma_y\right\|_2^2$ is viewed as a regularizer essential for the robustness of the controller in the presence of noises~\cite{dorfler2022bridging}.
      In different versions of DeePC, the regularizer form has been chosen to be $l_1$-norm~\cite{coulson2019data,wegner_data-enabled_2021}, squared $l_2$-norm~\cite{zhang_dimension_2023,carlet_real-time_2023}, or the hybrid of both~\cite{dorfler2022bridging}. Therefore, we adopt a general formulation that can cover the above cases.
   \end{itemize}
\end{problem}

Since the collected I/O data are contaminated with noise, a large amount of data is required for an accurate non-parametric representation of the system dynamics. However, the dimension $M$ of the auxiliary optimization variable $g$ grows linearly with respect to the amount of data, which poses a computational challenge to the online optimization of DeePC.
\section{COMPUTATIONALLY EFFICIENT DEEPC VIA LEARNING-BASED APPROXIMATION}

We propose a computationally efficient approximate DeePC by introducing the notion of the scoring function.
We first show that the DeePC objective can be viewed as minimizing the sum of the control cost and the score, and the evaluation of the latter is costly due to the scale of data matrix.
Therefore, we propose to approximate the scoring function with a reduced-order one via differentiable convex programming.
The parameters of the reduced scoring function are learned offline.
Finally, an approximate control law is formulated as minimizing the sum of the control cost and the learned approximate scoring function.
The overall framework is illustrated in Fig~\ref{fig:framework}.

\begin{figure}[!htbp]
   \centering
   \begin{tikzpicture}[>=Stealth,  squiggle/.style={-latex, line join=round, decorate, decoration={zigzag, segment length=4, amplitude=.9, post=lineto, post length=2pt}}]
      \node (min1) {$\mathrm{minimize}$};
      \node[right = 0.15 cm of min1, blue] (ell1) {$\ell(\tau)$};
      \node[above = 0.3 cm of ell1, blue] (cost) {Control Cost};
      \node[right = 0.3 cm of ell1] (plus1) {+};
      \node[right = 0.3 cm of plus1, red] (S) {$S(\tau)$};
      \node[above = 0.3 cm of S, red] (score) {Score};
      \node[right = of S] (DeePC) {DeePC};
      \node[below = of min1] (min2) {$\mathrm{minimize}$};
      \node[right = 0.15 cm of min2, blue] (ell2) {$\ell(\tau)$};
      \node[right = 0.3 cm of ell2] (plus2) {+};
      \node[right = 0.3 cm of plus2, red] (Shat) {$\hat{S}(\tau)$};
      \node[right = of Shat] (Ours) {Ours};
      \draw[blue] ($(ell1.south)-(0.03,0.25)$) -- ($(ell2.north)-(0.03,-0.25)$); 
      \draw[blue] ($(ell1.south)+(0.03,-0.25)$) -- ($(ell2.north)+(0.03,0.25)$); 
      \draw[->] (S) -- (DeePC);
      \draw[->, squiggle, red] (S) -- (Shat) node[midway, right, red] {Alg.~\ref{alg:score_learning}};
      \draw[->] (Shat) -- (Ours);
      \node[draw, dashed, fit=(ell1) (S) (min1)] (box) {};
      \node[draw, dashed, fit=(ell2) (Shat) (min2)] (box) {};
    \end{tikzpicture}
   \caption{Illustration of the overall framework. The optimization objective of DeePC can be split into two parts: the control cost $\ell(\tau)$ and the score function $S(\tau)$. Subsequently, $S(\tau)$ is approximated by $\hat{S}(\tau)$, which is learned offline. The learning-based efficient approximation of DeePC is formulated as minimizing the sum of the control cost $\ell(\tau)$ and the learned approximate score $\hat{S}(\tau)$.}
   \label{fig:framework}
   \vspace{-18pt}
\end{figure}
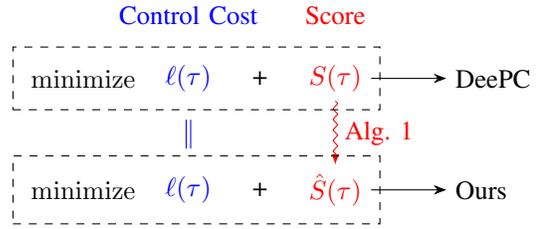

\subsection{DeePC through the Lens of Scoring Function}

Let $L=T_{\mathrm{ini}}+N$ be the length of the I/O trajectory sequence in the data matrix. We denote $$\tau \triangleq \operatorname{col}(\mathbf{u}, \mathbf{y})$$ as the predicted I/O trajectory sequence \revise{of length $L$}, where $\mathbf{u} \in \mathbb{R}^{mL}$ and $\mathbf{y} \in \mathbb{R}^{pL}$ are the predicted input and output sequences, respectively.

We propose to reformulate Problem~\ref{eq:DeePC} and separate the objective into two parts:
\begin{equation}
   \min _{\tau} \ell(\tau)+S(\tau),
\label{eq:DeePC_reformulated}
\end{equation}
where the objective functions are defined as
\revise{
\begin{equation}
	\ell(\tau)  = \|\mathbf{y} - \mathbf{r}\|_{\mathbf{Q}}^{2}+\|\mathbf{u}\|_{\mathbf{R}}^2 + \Ib_{\{=\tau_\mathrm{ini}\}}(E_{\mathrm{ini}}\tau) + \Ib_{\mathcal{U}^N \times \mathcal{Y}^N}(E_N\tau),
   \notag
   \vspace{-5pt}
\end{equation}
\begin{align}
   &S(\tau)  = \min _{g, \sigma_y}\lambda_{g1}\|g\|_1+\lambda_{g2}\|g\|_2^2+\lambda_{y1}\|\sigma_y\|_1+\lambda_{y2}\|\sigma_y\|_2^2 ,\nonumber \\
   & \qquad  \text{s.t. } \Hc g + E_{\sigma} \sigma_y = \tau, \label{eq:score}
\end{align}
and $\mathbf{r}\triangleq r\otimes \mathbf{1}_L$, $\mathbf{Q}\triangleq Q\otimes I_L$, $\mathbf{R}\triangleq R \otimes I_L$,
$E_{\mathrm{ini}}\triangleq\begin{bmatrix} I_{mT_{\mathrm{ini}}} & 0_{mT_{\mathrm{ini}}\times mN} & 0_{mT_{\mathrm{ini}}\times pT_{\mathrm{ini}}} & 0_{mT_{\mathrm{ini}}\times pN} \\ 0_{pT_{\mathrm{ini}}\times mT_{\mathrm{ini}}} & 0_{pT_{\mathrm{ini}}\times mN} & I_{pT_{\mathrm{ini}}} & 0_{pT_{\mathrm{ini}}\times pN} \end{bmatrix}$,
$E_N \triangleq\begin{bmatrix} 0_{mN\times mT_{\mathrm{ini}}} & I_{mN} & 0_{mN \times pT_{\mathrm{ini}}} & 0_{mN\times pN} \\ 0_{pN\times mT_{\mathrm{ini}}} & 0_{pN\times mN} &  0_{pN \times pT_{\mathrm{ini}}} & I_{pN} \end{bmatrix} $ are selection matrices that extract corresponding parts of the I/O sequence from $\tau$,
and $E_{\sigma} \triangleq \begin{bmatrix} 0_{pT_\mathrm{ini} \times mL} & -I_{pT_\mathrm{ini}} & 0_{pT_\mathrm{ini} \times pN} \end{bmatrix}^\top$.
In the notation of $\Ib_{\Uc^N\times \Yc^N}(\tau)$, ``$\times$" denotes Cartesian power and $\mathcal{U}^{N}$ denotes $N$-th Cartesian power of a set $\mathcal{U}$.
}

Here, $\ell(\cdot)$ and $S(\cdot)$ fulfill two orthogonal roles:

\begin{itemize}
\item $\ell(\cdot)$ represents the cost function of the control problem, which is distinct from and unaffected by the system dynamics.
\item $S(\cdot)$ is referred to as the \score. For any I/O trajectory $\tau
$, it evaluates the fitness of $\tau$ to the trajectory segments collected in the data matrix. It depends only on the system and is not influenced by any specific costs or constraints.
\end{itemize}

Note that $g, \sigma_y$ are auxiliary variables required for the evaluation of the scoring function, not part of the control policy,
and the dimension of $g$ grows with the number of trajectory segments in the data matrix.
Therefore, we desire to accelerate the evaluation of $S(\tau)$ using a learning-based approach which approximates the \score\ with fewer auxiliary variables.
We denote the approximate scoring function as $\hat{S}(\tau)$, then the overall control problem becomes:
\begin{equation}
   \min_{\tau} \ell(\tau)+\hat{S}(\tau).
   \label{eq:DeePC_approx}
\end{equation}

\subsection{Learning Objective}

Since the approximation goal is to make control input, i.e., the optimal solution, of the learned problem $\min\,l(\tau) + \hat{S}(\tau)$ close to that of $\min\,l(\tau) + S(\tau)$, we set the learning target to be minimizing the error between the \emph{proximal operators} of the approximate and true scoring functions $\hat{S}$ and $S$. 

The proximal operator~\cite{parikh2014proximal} of $\hat{S}$ is defined as:
\begin{equation}
   \operatorname{Prox}_{\hat{S}}(\tau)=\arg \min _{\hat\tau} \hat{S}(\hat\tau)+\frac{1}{2}\|\hat\tau-\tau\|^2.
   \label{eq:prox_def}
\end{equation}
Consequently, the optimality condition of~\eqref{eq:DeePC_approx}, i.e, $0 \in (\partial \ell + \partial \hat{S})(\hat{\tau}^\star)$, can be equivalently expressed with the proximal operators of $\ell$ and $\hat{S}$ using operator splitting methods. Therefore, the optimal solution $\hat{\tau}^\star$ to~\eqref{eq:DeePC_approx} can be determined once $\operatorname{Prox}_{\hat{S}}$ is given.

\subsection{Approximator Form: Differentiable Convex Program}

Given the learning target $\operatorname{Prox}_S$, the subsequent step is to select a parameterized family of $\hat{S}$ that satisfies the following conditions:
i) $\hat{S}(\tau)$ is capable of approximately representing the true \score\ $S(\tau)$;
ii) $\hat{S}$ is a convex function of $\tau$, and $\operatorname{Prox}_{\hat{S}}$ can be easily parameterized; iii) the control problem $\min_{\tau} \ell(\tau)+\hat{S}(\tau)$ should be computationally tractable.
Considering the above requirements, we adopt a differentiable convex programming approach to represent $\hat{S}$.
Specially, we adopt the following form:
\begin{subequations}
   \begin{align}
   \hat{S}(\tau)=
   \min_{z \in \mathbb{R}^{n_z}} & \|\operatorname{diag}(d_1)z\|_1 + \|\operatorname{diag}(d_2)z\|_2^2 \\
   \mathrm{s.t.}\ & Gz + W\tau = 0.
   \end{align}
\label{eq:score_approx}
\end{subequations}
Here, $\hat{S}$ is a mapping from an I/O sequence $\tau$ to a scalar standing for the optimal value of a convex optimization problem with $n_z$ variables and $m_z$ constraints, whose decision variable is $z$ and whose parameters include $\tau$ as well as learnable parameters $d_1 \in \mathbb{R}^{n_z}, d_2\in \mathbb{R}^{n_z}, G \in \mathbb{R}^{m_z \times n_z}, W \in \mathbb{R}^{m_z \times L(m+p)}$. To train the differentiable convex program means to update the parameters $d_1, d_2, G, W$ via gradient-based methods, such that the learning target ($\operatorname{Prox}_{\hat{S}} \approx \operatorname{Prox}_S$ in our case) is achieved. The size of the convex program, i.e., $n_z$ and $m_z$, can be arbitrarily chosen according to a trade-off between the approximation accuracy and the computational cost, but once fixed, it is invariant against the amount of data used for training.

We first show that the chosen form of $\hat{S}$ has adequate representational power to approximate $S$:

\begin{proposition}
   With sufficiently large $n_z$ and $m_z$, the true $S(\tau)$ is representable by \eqref{eq:score_approx}.
\end{proposition}

\begin{proof}
   Let $d_1 = [\underbrace{\lambda_{g1}, \ldots, \lambda_{g1}}_{M\ \text{times}}, \underbrace{\lambda_{\sigma_{y1}}, \ldots, \lambda_{\sigma_{y1}}}_{pT_\mathrm{ini}\ \text{times}}]^\top$,
   $d_2 = [\underbrace{\sqrt{\lambda_{g2}}, \ldots, \sqrt{\lambda_{g2}}}_{M\ \text{times}}, \underbrace{\sqrt{\lambda_{\sigma_{y2}}}, \ldots, \sqrt{\lambda_{\sigma_{y2}}}}_{pT_\mathrm{ini}\ \text{ times}}]^\top$, \revise{$G = \begin{bmatrix} \Hc & E_{\sigma} \end{bmatrix}$, $W = -I_{(m+p)L}$}, then $\hat{S}(\tau)$ defined in~\eqref{eq:score_approx} is equivalent to $S(\tau)$. \revise{Therefore, the true $S(\tau)$ can be representable by \eqref{eq:score_approx} with $n_z \geq M + pT_\mathrm{ini}$ and $m_z \geq (m+p)L$.}
\end{proof}

We next show how to evaluate and learn the proximal operator of the function $\hat{S}$ given in~\eqref{eq:score_approx}. By definition~\eqref{eq:prox_def} of proximal operator, we have:
   \begin{align}
   &(z^\star, \hat{\tau}^\star) = \arg\min_{z, \hat\tau} \|\operatorname{diag}(d_1)z\|_1 + \|\operatorname{diag}(d_2)z\|_2^2+ \nonumber \\
   &\mkern100mu  \|\hat\tau - \tau\|^2/2,\quad  \mathrm{s.t.} \; Gz + W\hat\tau = 0. \nonumber \\
   & \operatorname{Prox}_{\hat{S}}(\tau) = \hat\tau^\star.
   \label{eq:prox}
   \end{align}

To evaluate~\eqref{eq:prox} and find the gradient of $\tau^\star$ with respect to the learnable parameters $d_1, d_2, G, W$, we adopt an algorithm unrolling~\cite{gregor2010learning,monga2021algorithm} approach, which has been proved efficient for controller learning~\cite{lu2023bridging}.
In particular, we unroll the Douglas-Rachford Splitting (DRS)~\cite{eckstein1992douglas}, which, when applied to the problem~\eqref{eq:prox}, yields the following iterations:
\begin{subequations}
   \begin{align}
      &\left[\begin{array}{c}z^{k+\frac{1}{2}} \\ \hat\tau^{k+\frac{1}{2}}\end{array}\right] = \left[\begin{array}{c}\operatorname{sh}_{d_1, d_2}(\xi^k)\\ \frac{{\tau} + \eta^k}{2}\end{array}\right], \\
      &\left[\begin{array}{c}z^{k+1} \\ \hat\tau^{k+1}\end{array}\right] = \left(I - \tilde{G}^{\dagger}\tilde{G}\right)\left[\begin{array}{c}2z^{k+\frac{1}{2}}-\xi^k \\ 2\hat\tau^{k+\frac{1}{2}}-\eta^k\end{array}\right], \\
      &\left[\begin{array}{c}\xi^{k+1} \\ \eta^{k+1}\end{array}\right] = \left[\begin{array}{c}\xi^k+z^{k+1}-z^{k+\frac{1}{2}}\\ \eta^k + \hat\tau^{k+1}-\hat\tau^{k+\frac{1}{2}}\end{array}\right],
   \end{align}
   \label{eq:DRS}
\end{subequations}
where $\xi, \eta$ are auxiliary variables, $\tilde{G} = \begin{bmatrix} G & W \end{bmatrix}$, and the soft-thresholding operator $\operatorname{sh}_{d_1, d_2}(x)$ is defined as:
\begin{equation}
   \operatorname{sh}_{d_1, d_2}(x)_i= \begin{cases}\frac{x_i-\left|d_{1 i}\right|}{1+2 \left(d_{2 i}\right)^2} & \text { if } x_i-\left|d_{1 i}\right|>0, \\
   \frac{x_i+\left|d_{1 i}\right|}{1+2 \left(d_{2 i}\right)^2} & \text { if } x_i+\left|d_{1 i}\right|<0, \\
   0 & \text { otherwise }.\end{cases}
   \label{eq:sh}
\end{equation}
Here, $d_{1 i}$ and $d_{2 i}$ are the $i$-th elements of $d_1$ and $d_2$, respectively.

The unrolled iterations resemble a deep neural network that interleaves affine transformations with soft-thresholding nonlinearities, whose weights are functions of the learnable parameters $d_1, d_2, G, W$. This resemblance, visualized in Fig~\ref{fig:score_approx}, facilitates the training of the learnable parameters using established deep learning frameworks.

\begin{figure}[!htbp]
   \centering
   \includegraphics{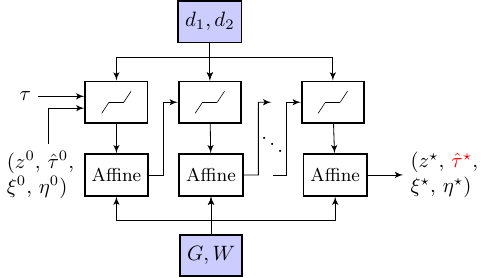}
   \caption{Illustration of evaluating $\operatorname{Prox}_{\hat{S}}$ with unrolled DRS iterations.  The input is an I/O sequence $\tau$ and the initial values of $z, \tau, \xi, \eta$, and the output is $\operatorname{Prox}_{\hat{S}}({\tau})$, denoted as $\hat\tau^{\star}$. The learnable parameters are $d_1, d_2, G, W$. \includegraphics*[height=0.2cm]{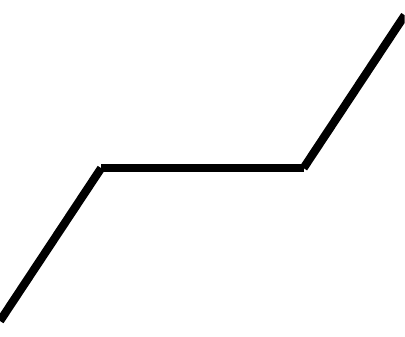} is the soft-thresholding operator~\eqref{eq:sh}, and the Affine block is the composition of all affine transformations in each iteration~\eqref{eq:DRS}.}
   \label{fig:score_approx}
   \vspace{-20pt}
\end{figure}

\subsection{Learning Algorithm}

Given a data matrix $\mathcal{H}$ consisting of I/O trajectory sequences, we first collect a dataset $\mathcal{D}$ of I/O trajectory sequences by linear combination of the trajectories in $\mathcal{H}$ and adding random noise to the trajectories to increase the diversity of the dataset. Then, we compute the ground truth $\operatorname{Prox}_{S}(\tau)$ for each trajectory $\tau$ in $\mathcal{D}$. Finally, we learn the parameters $\theta = (d_1, d_2, G, W$) of the approximate scoring function $\hat{S}$ by minimizing the mean squared error between $\operatorname{Prox}_{\hat{S}}$ and the ground truth $\operatorname{Prox}_{S}$ across all entries in the dataset $\mathcal{D}$.

The learning procedure is summarized in Algorithm~\ref{alg:score_learning}.
\vspace{-10pt}

\begin{algorithm} 
   \caption{Framework for Learning the Scoring Function}
   \label{alg:score_learning}
   \DontPrintSemicolon
   \LinesNumbered
   \KwIn{Data matrix $\mathcal{H}$ consisting of I/O trajectory sequences, regularization parameters $\lambda_{g1}, \lambda_{g2}, \lambda_{y1}, \lambda_{y2}$}
   \KwOut{Approximate scoring function parameters $\theta = (d_1, d_2, G, W)$}
   Collect a dataset $\mathcal{D}$ of I/O trajectory sequences by applying different noises to the trajectories in $\mathcal{H}$.\;
   Compute $\operatorname{Prox}_{S}(\tau)$ for each trajectory $\tau$ in $\mathcal{D}$.\;
   \revise{Randomly} initialize the parameters $\theta = (d_1, d_2, G, W)$.\;
   \For{epoch $= 1,2,\dots$}{
     Solve Problem~\ref{eq:prox} to obtain $\operatorname{Prox}_{\hat{S}}(\tau)$ for each trajectory $\tau$ in $\mathcal{D}$.\;
     Compute the mean squared error denoted by $\ell_{mse}(\theta; \mathcal{D}) = \sum_{\tau \in \mathcal{D}}\|\operatorname{Prox}_{S}(\tau)-\operatorname{Prox}_{\hat{S}}(\tau)\|_2^2$.\;
     Update $\theta$ according to \revise{the gradient of} the loss function $\ell_{mse}(\theta; \mathcal{D})$.\;
   }
\end{algorithm} 
\vspace{-10pt}

Note that learning the approximate scoring function is an offline process, and the learned parameters $\theta$ are used to substitute the original scoring function $S(\tau)$ to solve the overall control problem~\eqref{eq:DeePC_approx} online.

\section{SIMULATION}

In this section, we present the simulation results to demonstrate the effectiveness of our approach.
All computations are performed on a computer with an AMD Ryzen 9 5900X Processor clocked at 3.7GHz.
All controllers are implemented using the MOSEK optimizer.
Our code is available at \url{https://github.com/zhou-yh19/redpc}.

Consider a quadruple-tank system~\cite{johansson2000quadruple} whose simulation employs the same linearized system equations and system parameters as the LTI system used in~\cite{zhang_dimension_2023}:
\begin{align*}
&A = \begin{bmatrix}
   0.921 & 0 & 0.041 & 0 \\
   0 & 0.918 & 0 & 0.033 \\
   0 & 0 & 0.924 & 0 \\
   0 & 0 & 0 & 0.937
\end{bmatrix}, \\
&B = \begin{bmatrix}
   0.017 & 0.001 \\
   0.001 & 0.023 \\
   0 & 0.061 \\
   0.072 & 0
\end{bmatrix},
C = \begin{bmatrix}
   1 & 0 & 0 & 0 \\
   0 & 1 & 0 & 0
\end{bmatrix}.
\end{align*}

The system is subject to process noise $w(t) \sim \mathcal{N}(0, 0.01I_4)$ and measurement noise $v(t) \sim \mathcal{N}(0, 0.1I_2)$.

We follow the same setting as in~\cite{zhang_dimension_2023}, with the control cost matrices $Q = 35\cdot I_2$ and $R = 10^{-4}I_2$, the control input and output constraints $\mathcal{U} = \mathcal{Y} = [-2, 2]^2$, the setpoint $r(t) = \begin{bmatrix} 0.65,  0.77 \end{bmatrix}^T$, the prediction horizon $N = 20$, and the initial sequence length $T_\mathrm{ini} = 10$.

\emph{Data Collection:} 
\revise{The dataset is collected by simulating the system with random inputs generated by sampling from a uniform distribution over $\mathcal{U}$.
The corresponding outputs are recorded during the simulation.
In total, we collect 1500 I/O data points, which are subsequently organized into the data matrix \(\mathcal{H}\) in the form of Hankel matrix.}

We choose the regularization parameters of DeePC as $\lambda_{g1} = 1$, $\lambda_{g2} = 100$, $\lambda_{y1} = 100$, $\lambda_{y2} = 100000$.
And the approximate scoring function is learned with \revise{$m_z = 55, n_z=110$ and DRS iteration number 20}. The training procedure is completed within 260s on a single NVIDIA GeForce RTX 4090 GPU.
Fig~\ref{fig:trajectory} illustrates the first two states of the system under the control of both the DeePC controller and our proposed approach. The results demonstrate that both controllers track the setpoint effectively.

\vspace{-5pt}
\begin{figure}[!htbp]
   \centering
   \includegraphics{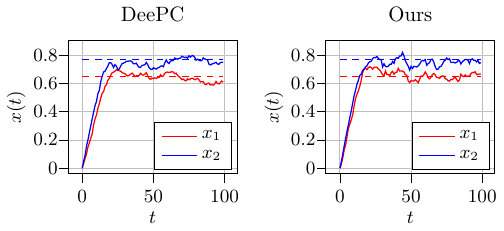}
   \caption{Tracking performance of DeePC and our approach.}
   \label{fig:trajectory}
\end{figure}
\vspace{-5pt}

In Table~\ref{tab:common}, we benchmark the control performance of our method with DeePC and MPC, detailing the sample average accumulative cost $\frac{1}{K}\sum_{k}\sum_t\left(\|y(t)-r(t)\|_Q^2+\|u(t)\|_R^2\right)$, the average and worst computation time per step.  Our approach achieves control performance comparable to MPC and DeePC, while significantly enhancing computational efficiency.

\begin{table}[!htbp]
   \vspace{-5pt}
   \begin{threeparttable}
   \centering
   \caption{Performance comparison of DeePC, MPC, and ours}\label{tab:common}
   \centering
   \setlength\tabcolsep{8pt}
   \begin{tabular}{c|c|cc}
   \toprule
   {\diagbox[dir=NW,innerwidth=2.2cm,height=2.2\line]{Metrics}{Method}} & MPC & DeePC & Ours \\
   \midrule
   Average Accumulative Cost & 305.00 & 290.68 & 296.25 \\
   Average Computation Time~(ms) & 96.35 & 153.71 & 30.41 \\
   Worst Computation Time~(ms) & 136.45 & 280.44 & 80.7 \\
   \bottomrule
   \end{tabular}
   \vspace{0.25em}
   {
   \footnotesize
   \revise{For the MPC, we assume the precise model of the system is available and update the control input with the true state at each step. To be consistent with the DeePC formulation, the MPC does not incorporate a terminal cost and the horizon is set to 20.}
   }
   \end{threeparttable}
\end{table}

\revise{Furthermore, we train the approximate scoring function of different sizes with the data matrix $\mathcal{H}$ and evaluate their performance on the control task. For comparison, we truncate the columns of $\mathcal{H}$ to adjust the number of I/O sequences used in DeePC. The results are shown in Fig.~\ref{fig:performance_curve}.
}
\begin{figure}[!htbp]
   \centering
   \includegraphics[scale=0.9]{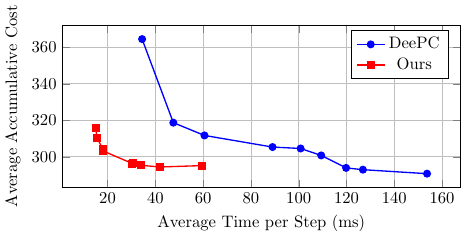}
   \caption{Comparison of the cost-time curves between DeePC and our approach. The average computation time for DeePC is influenced by the number of I/O sequences in the data matrix, whereas our approach's computation time primarily depends on the size of the learned approximate scoring function.}
   \label{fig:performance_curve}
   \vspace{-18pt}
\end{figure}

For DeePC, an increase in the number of I/O sequences in the data matrix results in better control performance but also longer computation time.
For our approach, the computation time is mainly determined by the size of the learned approximate scoring function, and there is also trade-off between control performance and computation time.
Nevertheless, it can be observed that under the same computational time constraints, our method achieves superior control performance compared to DeePC, and when aiming for comparable control quality, our approach significantly reduces the computation time. This demonstrates the effectiveness of our approach in enhancing the computational efficiency of DeePC while maintaining control performance. 
\section{CONCLUSIONS}
In this paper, we propose a computationally efficient approximate DeePC by introducing the notion of the scoring function and replacing it with a learned approximate one via differentiable convex programming. The parameters of the reduced scoring function are learned offline, and the control law is formulated as minimizing the sum of the control cost and the learned approximate scoring function. Simulation results demonstrate our approach can achieve a comparable tracking performance to DeePC while significantly reducing the computation time. In the future, we will apply our approach to the real-world control problems and investigate the generalization ability of the learned approximate scoring function.
\vspace{-1pt}
\bibliographystyle{IEEEtran}
\bibliography{ref}

\end{document}